\documentclass{revtex4}
\usepackage{amsmath,amsthm}
\usepackage{rotating}
\usepackage{multirow}
\usepackage{graphicx}

\usepackage{amsfonts}
\newtheorem{theorem}{Theorem}[section]

\newtheorem{proposition}[theorem]{Proposition}

\theoremstyle{remark}

\theoremstyle{definition}
\newtheorem{definition}[theorem]{Definition}

\theoremstyle{example}
\newtheorem{example}[theorem]{Example}

\theoremstyle{notation}

\newcommand{\bra}[1]{\langle#1|}
\newcommand{\ket}[1]{|#1\rangle}

\begin{document}

\title{The Choquet integral as an approximation to  density matrices with incomplete information}            
\author{A. Vourdas}
\affiliation{Department of Computer Science,\\
University of Bradford, \\
Bradford BD7 1DP, United Kingdom\\a.vourdas@bradford.ac.uk}

\begin{abstract}
A total set of $n$ states $\ket{i}$ and the corresponding projectors $\Pi(i)=\ket{i} \bra{i}$  are considered,
in a quantum system with $d$-dimensional Hilbert space $H(d)$.
A partially known density matrix $\rho$ with given $p(i)={\rm Tr}[\rho \Pi(i)]$  (where $i=1,...,n$ and $d\le n\le d^2-1$) is
considered, and its ranking permutation is defined. It is used to calculate the Choquet integral ${\cal C}(\rho)$ which is a positive semi-definite Hermitian matrix.
Comonotonicity is an important concept in the formalism, which is used to formalise the vague concept of
physically similar density matrices. It is shown that
${\cal C}(\rho)/{\rm Tr}[{\cal C}(\rho)]$ is  a density matrix which is a good approximation to the partially known density matrix $\rho$.
\end{abstract}
\maketitle

\section{Introduction}

The basic property of Kolmogorov probabilities is additivity. For exclusive events ($A\cap B=\emptyset$) we get $\mu(A\cup B)=\mu(A)+\mu(B)$.
Choquet \cite{C} introduced the weaker concept of capacities or non-additive probabilities.
They formalize the added value in an aggregation, where the `whole is greater than the sum of its parts'.
For example, if  two political parties get $10\%$ and $15\%$ of the votes, their coalition might get more (or less) than the sum of $25\%$ of the votes. 
Another example is two workers who produce $N_1, N_2$ items of a product per hour when they work individually, if they collaborate they  might produce more (or less) than $N_1+N_2$ items.
Lack of additivity requires a revision of the concept of integration, and this leads to Choquet integration,
which has been used extensively in Artificial Intelligence, Operations Research, Mathematical Economics, etc (e.g., \cite{D1,D2,D3,D4,D5}).

Choquet integrals are described with cumulative functions. In additive probabilities  the derivatives of cumulative functions are equal to the probability distributions.
In non-additive probabilities (capacities),  the derivatives of cumulative functions are in general different from the probability distributions.
The Choquet integral is a weighted average, with weights  the derivatives of cumulative functions.
This takes into account that an aggregation is different from the sum of its parts.
We note that in order to define a cumulative function, we need to have an order of the various alternatives.
In the finite case considered here, there are many orders and a choice needs to be made (as discussed below).

In a quantum context Choquet integrals have been used in \cite{C1,C2}. 
Here the need for  non-additive probabilities is seen with the following argument. Let $v_1, v_2$ be two vectors and $\Pi_1, \Pi_2$ the corresponding projectors.
Also let $\Pi_{12}$ be the projector to the two-dimensional space spanned by the $v_1,v_2$.
Then in general $\Pi_{12}\ne \Pi_1+\Pi_2$ (only for orthogonal vectors this is equality).
Projectors are intimately related to quantum probabilities, and this shows the need for the use of theories based on non-additive probabilities in a quantum context.

In this paper we consider a quantum system with $d$-dimensional Hilbert space $H(d)$.
In it we select a  total set of $n$ states $\ket{i}$ ($i=1,...,n$) and the corresponding projectors $\Pi(i)=\ket{i} \bra{i}$.
Let $\rho$ be an unknown density matrix with given $p(i)={\rm Tr}[\rho \Pi(i)]$  (where $i=1,...,n$ and $d\le n\le d^2-1$).
In this case we have incomplete information about $\rho$.
The corresponding Choquet integral ${\cal C}(\rho)$ is defined and is shown to be a positive semi-definite matrix.
The ${\cal C}(\rho)/{\rm Tr}[{\cal C}(\rho)]$ is shown to be a good (in the sense defined later) approximation to the partially known density matrix $\rho$.

Dealing with incomplete data is an important problem in both classical and quantum physics. 
An important approach that has been studied extensively in a classical context is the maximum entropy method (introduced in \cite{E1}). 
Various entropies have also been used for a similar purpose in a quantum context (e.g., \cite{E2,E3,E4,E5}).
In this paper we use a very different technique based on Choquet integration.

An important concept related to Choquet integrals, is comonotonicity of two density matrices $\rho _1, \rho _2$. There is added value
in an aggregation of components with different properties,  
because the various components play complementary role to each other.
In this case the whole is different from the sum of its parts, and the ${\cal C}[\lambda\rho _1+(1-\lambda) \rho _2]$ is different from $\lambda{\cal C}(\rho _1)+(1-\lambda){\cal C}(\rho_2)$.
But if the components of an aggregation have similar properties, this complementarity and added value are missing, the whole is equal to the sum of its parts, and
${\cal C}[\lambda\rho _1+(1-\lambda) \rho _2]$ is equal to $\lambda{\cal C}(\rho _1)+(1-\lambda){\cal C}(\rho_2)$.
Comonotonicity defines rigorously the intuitive concept of physically similar density matrices.

In section 2, we introduce capacities and Choquet integrals in a classical context.
There is much literature on these concepts in other than Physics areas,
and here we present briefly the concepts that we are going to bring into Quantum Physics.
In section 3, we introduce technical details (cumulative  projectors and their discrete derivatives,
M\"obius operators, etc) which are needed in the calculation of the Choquet integral.

In section 4, we introduce the Choquet integral ${\cal C}(\rho)$ of a density matrix $\rho$, and study its properties.
In section 5, we introduce  comonotonicity which is an important concept within Choquet integration.
Comonotonicity splits the set of density matrices into equivalence classes, and within each of these classes we define a total preorder.

In section 6, we explain that the ${\cal C}(\rho)/{\rm Tr}[{\cal C}(\rho)]$ is a good approximation to the partially known density matrix $\rho$, in the sense that they are weakly comonotonic to each other.
We conclude in section 7 with a discussion of our results.

\section{Capacities and Choquet integrals in a classical context}
\subsection{Capacities  or non-additive probabilities}

Let $\Omega$ be the set of all possible outcomes of an experiment (sample space).
Kolmogorov probability is a map $\mu$ from subsets of $\Omega$ (events) to $[0,1]$, such that $\mu (\emptyset )=0$, and also
\begin{eqnarray}\label{1}
\mu(A\cup B)-\mu(A)-\mu(B)+\mu(A\cap B)=0;\;\;\;A,B \subseteq \Omega.
\end{eqnarray}
In the case $A\cap B=\emptyset$ this reduces to the additivity relation
\begin{eqnarray}\label{B}
A\cap B=\emptyset\;\;\rightarrow\;\;\mu (A\cup B)= \mu (A)+\mu (B).
\end{eqnarray}
We might also have a normalization condition  $\mu (\Omega)=1$, but it is not essential.

Capacity or nonadditive probability, replaces Eq.(\ref{1}) with  the weaker relation
\begin{eqnarray}\label{A}
A \subseteq B\;\rightarrow\;\mu (A)\le \mu (B)
\end{eqnarray}
In this case the additivity relation of Eq.(\ref{B}) does not hold in general.

\subsection{Cumulative functions and their discrete derivatives}

We consider the case where $\Omega$ is the finite set
\begin{eqnarray}
\Omega=\{1,...,n\}.
\end{eqnarray}
We note that due to the non-additivity property, in general 
\begin{eqnarray}\label{bnn}
\sum _{i=1}^n \mu(\{i\})\ne 1.
\end{eqnarray}
We consider a permutation $\sigma$ of the elements of $\Omega$:
\begin{eqnarray}\label{1A}
\sigma(1),\sigma(2),...,\sigma(n).
\end{eqnarray}
Then the 
\begin{eqnarray}\label{1B}
&&m_\sigma (n)=\mu[\{\sigma(n)\}],\nonumber\\
&&m_\sigma(n-1)= \mu[\{\sigma(n-1), \sigma(n)\}]\nonumber\\
&&........... \nonumber\\
&&m_\sigma (1)=\mu [\{\sigma(1), ...,\sigma(n)\}]=1.
\end{eqnarray}
is cumulative function (starting from the right which is the convention in this area), and depends on the permutation $\sigma$.
Since the set $\Omega$ is finite, there is no `natural ordering' of its elements, and the cumulative function depends on the ordering $\sigma$ of its elements.
The number of possible orderings is $n!$.

We also define the `discrete derivative' of the cumulative function, which also depends on $\sigma$:
\begin{eqnarray}\label{vb}
\nu _\sigma (i)&=&m_\sigma(i)-m_\sigma (i+1);\;\;\;\;i=1,...,n-1\nonumber\\
\nu _\sigma(n)&=&m_\sigma (n).
\end{eqnarray}
It is easily seen that
\begin{eqnarray}\label{pup}
\sum_{i=1}^n\nu _\sigma(i)=1.
\end{eqnarray}
We note that for additive probabilities  $\nu_\sigma(i)= \mu[\sigma (i)]$,
but for non-additive capacities this is not true.
Also for non-additive probabilities the sum of $\mu(i)$ is in general different from $1$ (Eq.(\ref{bnn})), while the sum of  $\nu_\sigma(i)$ is equal to $1$ (Eq.(\ref{pup})).

\subsection{Choquet integrals}

The starting point of integration is the additivity relation in Eq.(\ref{B}).
Choquet revised the concept of integration for capacities.

We consider a function $f$ on the finite set $\Omega$, which takes the  real non-negative values $f(1),...,f(n)$.
We call ${\cal L}_n$ the set of these functions.
We relabel this function, using a ranking  permutation $\tau$ of the indices, so that 
\begin{eqnarray}\label{234}
f[\tau (1)]\le ...\le f[\tau (n)].
\end{eqnarray}
The Choquet integral of $f$ with respect to the capacity $\mu$ is given by
\begin{eqnarray}\label{87}
{\cal C}(f;\mu)&=&\sum _{i=1}^{n} f[\tau (i)]\nu _\tau (i)\nonumber\\&=&f[\tau (1)]m_\tau (1)+\{f[\tau (2)]-f[\tau(1)]\}m_\tau (2)+...+
\{f[\tau (n)]-f[\tau(n-1)]\}m_\tau (n)
\end{eqnarray}
We use here the discrete derivatives  $\nu _\tau (i)$ of the cumulative functions, related to the ranking permutation $\tau$ of the function $f$.
In a weighted average we multiply the values of a function with the corresponding probabilities.
In a Choquet integral we replace the probabilities with discrete derivatives  of cumulative functions. 

For Kolmogorov probabilities the derivatives of cumulative functions are  probability distributions, or in the discrete case
discrete derivatives  of cumulative functions are equal to the probabilities.
This is no longer true for capacities, and the distinction between the two is important for Choquet integration.

\subsection{Comonotonic functions}\label{nnmm}
 In general two functions have different ranking and  consequently
\begin{eqnarray}
{\cal C}(f+g;\mu)\ne {\cal C}(f;\mu)+{\cal C}(g;\mu).
\end{eqnarray}

Comonotonic functions\cite{D1,D2,D3,D4,D5}, are functions with the same ranking permutation $\tau$. 
In this case the $\nu_\tau (i)$,  are the same for the two functions. Therefore
for comonotonic functions $f,g$ and positive $a,b$
\begin{eqnarray}\label{301}
{\cal C}(af+bg;\mu)= a{\cal C}(f;\mu)+b{\cal C}(g;\mu).
\end{eqnarray}

A constant function $c$ is comonotonic with any function $f$ and therefore
\begin{eqnarray}
{\cal C}(f+c;\mu)={\cal C}(f;\mu)+c.
\end{eqnarray}

Comonotonicity is not transitive in the set ${\cal L}_n$. 
For example, the function $f(i)=1$ is comonotonic to every function $g(i)$ and yet there are functions $g_1(i), g_2(i)$ which are not comonotonic.
Transitivity fails by degeneracies, i.e., equalities in Eq.(\ref{234}).
We define a subset of ${\cal L}_n$ where comonotonicity is transitive.
\begin{definition}
${\cal L}_n^\prime $ is a subset of ${\cal L}_n$ which contains functions such that $f(i)\ne f(j)$ for all $i,j=1,...,n$ (i.e., we have only strict inequalities in Eq.(\ref{234})).
\end{definition}
Comonotonicity is transitive in  ${\cal L}_n^\prime $.
Consequently, comonotonicity is an equivalence relation in ${\cal L}_n^\prime $, and partitions it into $n!$ equivalence classes. 
Each equivalence class has its own ranking permutation and its own set of discrete derivatives of the cumulative functions $\{\nu _\sigma (i)\;|\;i=1,...,n\}$.
The set ${\cal L}_n \setminus{\cal L}_n^\prime $ contains the `boundary functions' (with equalities in Eq.(\ref{234})).

\subsection{The Choquet integral in terms of M\"obius transforms of the capacities}\label{MOB}

The M\"obius transform is used extensively in Combinatorics, 
after the work by Rota\cite{R}. It is a generalization of
the inclusion-exclusion principle that gives the cardinality of the union of overlaping sets.
Rota generalized this to partially ordered structures.

The M\"obius transform of capacities leads to the following function ${\mathfrak d} (A)$ (where $A\subseteq \Omega$)\cite{VO2}:
\begin{eqnarray}\label{M}
{\mathfrak d} (A)=\sum  _{B\subseteq A}(-1)^{|A|-|B|} \mu(B),
\end{eqnarray}
where $|A|$, $|B|$ are the cardinalities of these sets.
For example, if $A=\{a_1,...,a_m\}$, then
\begin{eqnarray}
{\mathfrak d}(\{a_i\})&=&\mu (\{a_i\});\;\;\;\;\;
{\mathfrak d} (\{a_i,a_j\})=\mu (\{a_i,a_j\})-\mu(\{a_i\})-\mu(\{a_j\})
\end{eqnarray}
etc.
The inverse M\"obius transform is
\begin{eqnarray}\label{b7}
\mu (A)=\sum _{B\subseteq A}{\mathfrak d} (B).
\end{eqnarray}

The following proposition expresses the Choquet integral in terms of  M\"obius transforms of the capacities, and is given without proof\cite{D4,D5,C1}.
\begin{proposition}\label{PRO12}
The Choquet integral of Eq.(\ref{87}) is given in terms of the M\"obius transform ${\mathfrak d}$ of the capacities $\mu$, as
\begin{eqnarray}\label{88}
{\cal C}(f; \mu)&=&\sum _{A\subseteq \Omega} {\mathfrak d} (A){\min}[f(A)].
\end{eqnarray}
For all subsets $A$ of $\Omega$, we multiply ${\mathfrak d} (A)$ with the minimum value of the function in this subset, and we add the results.
\end{proposition}

\subsection{Example}\label{EXA23}
Three students labelled $1,2,3$  collaborate on a project. The
\begin{eqnarray}
&&\mu(1)=0.05;\;\;\;\mu(2)=0.1;\;\;\;\mu(3)=0.15\nonumber\\
&&\mu(1,2)=0.2;\;\;\;\mu(1,3)=0.25;\;\;\;\mu(2,3)=0.2;\;\;\;\mu(1,2,3)=0.3
\end{eqnarray}
show the percentage of the project, that a collaboration of any subset of these students completes, per unit of time (e.g., one hour).
For example the student $2$ on his own completes $0.1$ of the project, the collaboration of students $1,2$ completes $0.2$ of the project, etc. 
It is clear that in a collaboration $\mu(i,j)\ne \mu(i)+\mu(j)$ in general, and this shows the need for capacities. 
For example, $\mu(1,2)> \mu(1)+\mu(2)$ because the collaboration of the students $1,2$ is very sucessful.
On the other hand, $\mu(2,3)< \mu(2)+\mu(3)$ because the collaboration of the students $2,3$ is not very sucessful.
In this example, the $\mu(1,2,3)=0.3$ (it is not normalized to $1$).

All three students start working on the project together at $t=0$, but at $t=\lambda$ (where $\lambda <3$) , student $2$ leaves and the remaining students $1,3$ continue to work together until $t=3$.
Then the student $3$ leaves and the remaning student $1$ works on his own until $t=5$. The number of hours that each student worked is
\begin{eqnarray}
f(1)=5;\;\;\;f(2)=\lambda;\;\;\;f(3)=3.
\end{eqnarray}
We will calculate what percentage of the project has been completed, using the Choquet integral.

Since $\lambda <3$, we have $f(2)<f(3)<f(1)$ and therefore the ranking permutation is
\begin{eqnarray}
\tau(1)=2;\;\;\;\tau(2)=3;\;\;\;\tau(3)=1.
\end{eqnarray}
As $\lambda $ varies with $\lambda <3$, we get various functions $f(i)$ which are comonotonic.

Then
\begin{eqnarray}
m_\tau (3)=\mu(1)=0.05;\;\;\;m_\tau (2)= \mu (1,3)=0.25;\;\;\;m_\tau (1)=\mu(1,2,3)=0.3,
\end{eqnarray}
and 
\begin{eqnarray}
{\cal C}=f[\tau (1)]m_\tau (1)+\{f[\tau (2)]-f[\tau(1)]\}m_\tau (2)+\{f[\tau (3)]-f[\tau(2)]\}m_\tau (3)=0.85+0.05\lambda.
\end{eqnarray}
Alternatively, we can calculate the
\begin{eqnarray}
\nu _\tau (1)=m_\tau(1)-m_\tau (2)=0.05;\;\;\;\nu _\tau (2)=m_\tau(2)-m_\tau (3)=0.2;\;\;\;\nu_\tau(3)=m_\tau(3)=0.05.
\end{eqnarray}
and then
\begin{eqnarray}
{\cal C}=f[\tau(1)]\nu _\tau (1)+f[\tau(2)]\nu _\tau (2)+f[\tau(3)]\nu _\tau (3)=0.85+0.05\lambda
\end{eqnarray}
The Choquet integral gives the percentage of the project that has been completed, as a function of $\lambda$.

This example should be compared and contrasted with the following example which involves the same students.
All three students start working on the project together at $t=0$, but at $t=\lambda$ (where $\lambda <3$) , student $3$ leaves and the remaining students $1,2$ continue to work together until $t=3$.
Then the student $2$ leaves and the remaning student $1$ works on his own until $t=5$. In this case
\begin{eqnarray}
f(1)=5;\;\;\;f(3)=\lambda;\;\;\;f(2)=3,
\end{eqnarray}
and $f(3)<f(2)<f(1)$. Therefore the ranking permutation is
\begin{eqnarray}
\sigma (1)=3;\;\;\;\sigma(2)=2;\;\;\;\sigma(3)=1.
\end{eqnarray}
Then
\begin{eqnarray}
m_\sigma (3)=\mu(1)=0.05;\;\;\;m_\sigma (2)= \mu (1,2)=0.2;\;\;\;m_\sigma (1)=\mu(1,2,3)=0.3,
\end{eqnarray}
and 
\begin{eqnarray}
{\cal C}=f[\sigma(1)]m_\sigma (1)+\{f[\sigma (2)]-f[\sigma(1)]\}m_\sigma (2)+\{f[\sigma (3)]-f[\sigma(2)]\}m_\sigma (3)=0.7+0.1\lambda.
\end{eqnarray}
Alternatively, we can calculate the
\begin{eqnarray}
\nu _\sigma (1)=m_\sigma(1)-m_\sigma (2)=0.1;\;\;\;\nu _\sigma (2)=m_\sigma(2)-m_\sigma (3)=0.15;\;\;\;\nu_\sigma(3)=m_\sigma(3)=0.05.
\end{eqnarray}
and then
\begin{eqnarray}
{\cal C}=f[\sigma(1)]\nu _\sigma (1)+f[\sigma(2)]\nu _\sigma (2)+f[\sigma(3)]\nu _\sigma (3)=0.7+0.1\lambda.
\end{eqnarray}

The two examples show the need for the ranking permutations and also for the comonotonic functions and the equivalence classes related to them.
In each example, as $\lambda$ varies with $\lambda < 3$, we get physically similar cases formalized with comonotonic functions and with the fact that they all belong to the same equivalence class.
But as we go from the first example to the second one, we get physically different cases formalized with change in the equivalence class.

\section{Cumulative projectors and their discrete derivatives}\label{ASD}

In a $d$-dimensional Hilbert space $H(d)$, we consider a reference set of $n\ge d$ states 
\begin{eqnarray}\label{sig}
\Sigma =\{\ket{i}\;|\;i\in \Omega\};\;\;\;\Omega=\{1,...,n\}
\end{eqnarray}
such that:
\begin{itemize}
\item
Any subset of $d$ of these states, are linearly independent.
\item
$\Sigma$ and also any of its subsets with $r\ge d$ of these states, are total sets. We recall, that a
 set of vectors in a Hilbert space is total, if there is no vector which is orthogonal to all vectors in the set.
\end{itemize}
We emphasize from the outset that there is redundancy in $\Sigma$ because it has more vectors than the dimension of the space.
Redundancy is essential in noisy situations, and this is indeed the merit of this approach\cite{C1}.

Let $h_1, h_2$ be two subspaces of $H(d)$. Their disjunction is
\begin{eqnarray}\label{V1}
h_1\vee h_2={\rm span}(h_1\cup h_2).
\end{eqnarray}
This is the quantum OR operation and includes all superpositions of vectors in the two spaces (unlike the Boolean OR which is simply the union of sets).
Their conjunction is the logical AND
\begin{eqnarray}\label{V2}
h_1\wedge h_2=h_1\cap h_2.
\end{eqnarray}

Let $H(\{i\})$ be the one-dimensional subspace that contains the vector $\ket{i}$, 
and $H(A)$ be the subspace spanned by all the states $\ket{i}$ with $i\in A\subseteq \Omega$:
\begin{eqnarray}
H(A)=\bigvee _{i\in A}H(\{i\}).
\end{eqnarray}
We call $\Pi(A)$ the projector to the subspace $H(A)$. In particular
\begin{eqnarray}
\Pi(\{i\})=\ket{i}\bra{i};\;\;\;\Pi(\emptyset)=0.
\end{eqnarray}
There are $2^n$ projectors $\Pi(A)$. 
If $|A|\ge d$ then $\Pi(A)={\bf 1}$, so some of these $2^n$ projectors are equal to ${\bf 1}$.
Also
\begin{eqnarray}
&&{\rm Tr}[\Pi(A)]=|A|\;\;{\rm if}\;\;|A|<d\nonumber\\
&&{\rm Tr}[\Pi(A)]=d\;\;{\rm if}\;\;|A|\ge d.
\end{eqnarray}

We note that
\begin{eqnarray}
H(A\cup B)=H(A)\vee H(B),
\end{eqnarray}
but in general
\begin{eqnarray}\label{2A}
\Pi(A\cup B)\ne \Pi(A)+ \Pi(B),
\end{eqnarray}
and in particular
\begin{eqnarray}\label{377}
\Pi(A)\ne \sum _{i\in A}\Pi(\{i\}).
\end{eqnarray}
Only if the kets $\ket{i}$ where $i\in A$ are orthogonal to each other, we get equality in this equation.
 The M\"obius operators below, quantify the difference between the two sides in equations like (\ref{2A}), (\ref{377}). 

The trace of the projectors $\Pi(A)$ times a density matrix, gives capacities (non-additive probabilities).
In this sense, the projectors $\Pi(A)$ are the quantum analogue of the capacities $\mu$ in the classical case.
We note that the $\Pi(A), \Pi(B)$ do not commute in general, and our apprach of using capacities and Choquet integrals, provides an alternative complementary methodology to non-commutativity. 

In practical calculations, the projectors $\Pi(A)$ can be calculated as follows. 
We express the vectors $\ket{i}$ where $i\in A$, as $d\times 1$ columns, and then write the $d\times |A|$ matrix ${\mathfrak A}$ which has as columns these vectors.
The projector $\Pi(A)$ is given by
\begin{eqnarray}\label{670A}
\Pi(A)={\mathfrak A}({\mathfrak A}^\dagger {\mathfrak A})^{-1}{\mathfrak A}^\dagger.
\end{eqnarray}

We next consider a permutation $\sigma$ of the elements of $\Omega$ (which is also  a permutation of the states in $\Sigma$) and in analogy to the cumulative function in Eq.(\ref{1B}) we introduce the cumulative projectors
(starting from the right):
\begin{eqnarray}\label{1C}
&&\pi_\sigma (n)=\Pi[\{\sigma(n)\}],\nonumber\\
&&\pi_\sigma(n-1)= \Pi[\{\sigma(n-1), \sigma(n)\}]\nonumber\\
&&........... \nonumber\\
&&\pi_\sigma (n-d+1)= \Pi[\{\sigma(n-d+1),...,\sigma(n)\}]={\bf 1}\nonumber\\
&&........... \nonumber\\
&&\pi_\sigma (1)=\Pi [\{\sigma(1), ...,\sigma(n)\}]={\bf 1}.
\end{eqnarray}
In analogy to Eq.(\ref{vb}) we also define the `discrete derivatives' of the cumulative projectors, which also depend on $\sigma$:
\begin{eqnarray}
\varpi _\sigma (i)&=&\pi_\sigma(i)-\pi_\sigma (i+1);\;\;\;\;i=1,...,n-1\nonumber\\
\varpi _\sigma(n)&=&\pi_\sigma (n).
\end{eqnarray}
It is easily seen  that 
\begin{eqnarray}
\varpi _\sigma (i)=0;\;\;\;\;i=1,...,n-d.
\end{eqnarray}
and also that
\begin{eqnarray}
\sum_{i=n-d+1}^n\varpi _\sigma(i)=1;\;\;\;\varpi _\sigma(i)\varpi _\sigma(j)=\delta _{ij}\varpi _\sigma(i);\;\;\;i,j=n-d+1,....,n.
\end{eqnarray}
The $\varpi _\sigma(i)$ (with $i=n-d+1,....,n$) are a complete set of orthogonal projectors in $H(d)$.

\section{The Choquet integral in a quantum context}\label{BBB}

Let ${\cal M}_d$ be the set of all density matrices of quantum systems with Hilbert space $H(d)$.
For $\rho \in {\cal M}_d$, let
\begin{eqnarray}
 p(i)={\rm Tr}[\rho \Pi (\{i\})];\;\;\;i=1,...,n;\;\;\;d\le n\le d^2-1.
\end{eqnarray}
Each $p(i)$ is a probability. The various $p(i)$ can be measured using different ensembles described by the same density matrix $\rho$ (because the $ \Pi (\{i\})$ do not commute in general).

The following proposition gives bounds for the $\sum _{i=1}^n p(i)$.
\begin{proposition}\label{PRO38}
Let ${\cal E}_{\rm min}$, ${\cal E}_{\rm max}$ be the minimum and maximum eigenvalues of the density matrix
\begin{eqnarray}
Q=\frac{1}{n}\sum _{i=1}^n\Pi (\{i\}).
\end{eqnarray}
Then
\begin{eqnarray}\label{BB}
{\cal E}_{\rm max}\ge \frac{1}{n}\sum _{i=1}^n p(i)\ge {\cal E}_{\rm min}.
\end{eqnarray}
\end{proposition}
\begin{proof}
$Q$ is a $d\times d$ positive semi-definite  matrix with trace $1$, i.e., a density matrix.
Let ${\cal E}_j$, $q_j$ ($j=1,...,d$) be the eigenvalues and eigenprojectors of $Q$.
\begin{eqnarray}
Q=\sum _{j=1}^d{\cal E}_jq_j;\;\;\;\sum _{j=1}^dq_j={\bf 1}.
\end{eqnarray}
The states in the set $\Sigma$ in Eq.(\ref{sig}) form a total set, and from this follows that  $\sum _{j=1}^dq_j={\bf 1}$.

Then for any density matrix
\begin{eqnarray}
\frac{1}{n}\sum _{i=1}^n p(i)={\rm Tr}(\rho Q)=\sum _{j=1}^d{\cal E}_j{\rm Tr}(\rho q_j)\ge {\cal E}_{\rm min}\sum _{j=1}^d{\rm Tr}(\rho q_j)={\cal E}_{\rm min}.
\end{eqnarray}
and
\begin{eqnarray}
\frac{1}{n}\sum _{i=1}^n p(i)={\rm Tr}(\rho Q)=\sum _{j=1}^d{\cal E}_j{\rm Tr}(\rho q_j)\le {\cal E}_{\rm max}\sum _{j=1}^d{\rm Tr}(\rho q_j)={\cal E}_{\rm max}.
\end{eqnarray}

\end{proof}

A density matrix has $d^2-1$ degrees of freedom, and therefore the set $\{p(i)\}$ with $i=1,...,n$ and $d\le n\le d^2-1$ contains only part of the information in $\rho$.
We consider such a  partially known density matrix $\rho$, with given set $\{p(i)\}$.
If $\rho=(\rho _{\alpha \beta})$ and  $\Pi (\{i\})= (\Pi _{\alpha \beta} (\{i\}))$ in some basis, then
\begin{eqnarray}\label{rr}
 p(i)=\sum  _{\alpha , \beta} \rho _{\alpha \beta}\Pi _{\beta \alpha} (\{i\});\;\;\;
\rho _{\alpha \beta}=\rho _{\beta \alpha}^*;\;\;\;\sum \rho _{\alpha \alpha}=1,
\end{eqnarray}
is a system of $n$ equations with $d^2-1$ unknowns (the  $\rho _{\alpha \beta}$ written in terms of real variables).
Therefore $d^2-1-n$ real variables in $\rho$ are chosen arbitrarily and the rest $n$ real variables are found by solving the system of Eqs.(\ref{rr}). 
This algorithm leads to a family of Hermitian matrices $\rho$ with trace $1$, but only the positive semi-definite ones are density matrices.
The process of filtering out which of these matrices are positive semi-definite can be computationally expensive.
Below we introduce the Choquet integral which is a  positive semi-definite matrix and can provide an approximation to the density matrices that  have a given set $\{p(i)\}$.

For a given set of projectors $\{\Pi (\{i\})\}$, we use the notation $\rho _1\overset{p}{\sim} \rho _2$ if the density matrices $\rho _1, \rho_2$ have the same set $\{p(i)\}$.
 $\overset{p}{\sim}$ is an equivalence relation and 
we denote with ${\mathfrak M}_d[p(i)]$ the equivalence class of all density matrices with the same set of measured probabilities $\{p(i)\}$.
The set of density matrices ${\cal M}_d$ is partitioned as
\begin{eqnarray}
{\cal M}_d=\bigcup _{\{p(i)\}}{\mathfrak M}_d[p(i)]
\end{eqnarray}

\subsection{The ranking permutation of a density matrix: its `postcode' in the Hilbert space}
\begin{definition}
Let $\rho$ be a density matrix. We relabel the $p(i)={\rm Tr}[\rho \Pi (\{i\})]$  as $p[\sigma(i)]$ ($i=1,...,n$), where $\sigma$ is a permutation, such that
\begin{eqnarray}\label{56}
0\le p[\sigma(1)]\le p[\sigma(2)]\le ....\le p[\sigma(n)] .
\end{eqnarray}
We call $\sigma$ the ranking permutation of the density matrix $\rho$.
\end{definition}
The $( \sigma(n),  \sigma(n-1),..., \sigma(1))$ is a kind of `postcode' of the density matrix $\rho$ in the Hilbert space, with respect to the reference set of vectors $\Sigma$ in Eq.(\ref{sig}).
The quantum states that $\rho$ describes, are located mainly in $H(\{\sigma(n)\})$ (with weight $ p[\sigma(n)]$), less in $H(\{\sigma(n-1)\})$  (with weight $ p[\sigma(n-1)]$), even less in  $H(\{\sigma(n-2)\})$  (with weight $ p[\sigma(n-2)]$), etc.
A quantum state is in general in all  (non-orthogonal)  subspaces $H(\{\sigma(n)\})$, $H(\{\sigma(n-1)\})$, $H(\{\sigma(n-2)\})$, etc, with different weights.

\subsection{The Choquet integral ${\cal C}(\rho)$: a positive semi-definite  operator}

In analogy to Eq.(\ref{87}), we introduce the Choquet integral
\begin{eqnarray}\label{010}
{\cal C}(\rho)=\sum_{i=n-d+1}^{n}p [\sigma(i)]\varpi _\sigma (i),
\end{eqnarray}
where $\sigma$ is the ranking permutation of $\rho$.
${\cal C}(\rho)$ is a positive semi-definite Hermitian operator with eigenvalues the $d$ largest values of $p(i)$, and eigenprojectors the $\varpi _\sigma (i)$.

Only the probabilities $p(i)$ enter in this calculation, and therefore all density matrices in the equivalence class ${\mathfrak M}_d[p(i)]$ have the same Choquet integral ${\cal C}(\rho)$.
In fact, only the `important' values of $p(i)$ (the $d$ largest values) enter in the integral.
Therefore two different equivalence classes with the same $d$ largest values of $p(i)$, have the same Choquet integral ${\cal C}(\rho)$.

An alternative equivalent expression for the Choquet integral is 
\begin{eqnarray}\label{00}
{\cal C}(\rho)&=&p [\sigma(n-d+1)]{\bf 1}+\{p[\sigma (n-d+2)]-p [\sigma (n-d+1)]\}\pi_\sigma(n-d+2)+....\nonumber\\
&+&\{p[\sigma (n)]-p [\sigma (n-1)]\}\pi_\sigma (n)
\end{eqnarray}
Eqs.(\ref{010}), (\ref{00}) are analogous to Eq.(\ref{87}).

If two of the values of the $p(i)$ function are equal to each other ($p [\sigma(i)]=p [\sigma(i+1)]$), there are two different orderings but they both lead to the same Choquet integral.
Indeed, the contribution of these two terms  is
\begin{eqnarray}\label{ppp}
p [\sigma(i)]\varpi _\sigma (i)+p [\sigma(i+1)]\varpi _\sigma (i+1)=p [\sigma(i)]\{\Pi[\{\sigma(i),...,\sigma(n)\}]-\Pi[\{\sigma(i+2),...,\sigma(n)\}]\}.
\end{eqnarray}
This result does not change if we swap the two states labelled with $i$ and $i+1$.

\begin{example}

If the $d$ largest values of  $p(i)={\rm Tr}[\rho \Pi (\{i\})]$ are equal to each other and equal to $p$, then Eq.(\ref{010}) gives
\begin{eqnarray}\label{AA}
{\cal C}(\rho)=p{\bf 1}.
\end{eqnarray}
An example of this is the case
\begin{eqnarray}\label{AA}
\rho=\frac{1}{d}{\bf 1}.
\end{eqnarray}
Then $p(i)=\frac{1}{d}$, and 
\begin{eqnarray}\label{AA}
{\cal C}\left (\frac{1}{d}{\bf 1}\right )=\frac{1}{d}{\bf 1}.
\end{eqnarray}
\end{example}

\subsection{The trace of the Choquet integral ${\cal C}(\rho)$ }

The following proposition is similar to  proposition \ref{PRO38} and it provides bounds for the ${\rm Tr}[{\cal C}(\rho)]$.
\begin{proposition}
Let $\rho$ be a density matrix with ranking permutation $\sigma$. Also let ${\cal E}_{\rm min}(\sigma)$, ${\cal E}_{\rm max}(\sigma)$ be the minimum and maximum eigenvalues of the density matrix
\begin{eqnarray}
Q_{\sigma}=\frac{1}{d}\sum _{i=n-d+1}^n\Pi [\{\sigma(i)\}].
\end{eqnarray}
Then
\begin{eqnarray}\label{BBB}
{\cal E}_{\rm max}(\sigma)\ge \frac{1}{d}{\rm Tr}[{\cal C}(\rho)]\ge {\cal E}_{\rm min}(\sigma).
\end{eqnarray}
\end{proposition}
\begin{proof}
$Q_{\sigma}$ is a $d\times d$ positive semi-definite  matrix with trace $1$, i.e., a density matrix.
Let ${\cal E}_j(\sigma)$, $q_j(\sigma)$ ($j=1,...,d$) be the eigenvalues and eigenprojectors of $Q_{\sigma}$:
\begin{eqnarray}
Q_{\sigma}=\sum _{j=1}^d{\cal E}_j(\sigma)q_j(\sigma);\;\;\;\sum _{j=1}^dq_j(\sigma)={\bf 1}.
\end{eqnarray}
Any $d$ of the states in the set $\Sigma$ in Eq.(\ref{sig}) form a total set, and from this follows that  $\sum _{j=1}^dq_j(\sigma)={\bf 1}$.

Then for any density matrix
\begin{eqnarray}
\frac{1}{d}{\rm Tr}[{\cal C}(\rho)]=\frac{1}{d}\sum_{i=n-d+1}^{n}p [\sigma(i)]=\sum _{j=1}^d{\cal E}_j(\sigma){\rm Tr}[\rho q_j(\sigma)]\ge {\cal E}_{\rm min}(\sigma)\sum _{j=1}^d{\rm Tr}[\rho q_j(\sigma)]={\cal E}_{\rm min}(\sigma).
\end{eqnarray}
\begin{eqnarray}
\frac{1}{d}{\rm Tr}[{\cal C}(\rho)]=\frac{1}{d}\sum_{i=n-d+1}^{n}p [\sigma(i)]=\sum _{j=1}^d{\cal E}_j(\sigma){\rm Tr}[\rho q_j(\sigma)]\le {\cal E}_{\rm max}(\sigma)\sum _{j=1}^d{\rm Tr}[\rho q_j(\sigma)]={\cal E}_{\rm max}(\sigma).
\end{eqnarray}

\end{proof}

The following proposition shows that ${\rm Tr}[{\cal C}(\rho)]$
 is a convex function. 
\begin{proposition}\label{1234}
The ${\rm Tr}[{\cal C}(\rho)]$
 is a convex function.
\end{proposition}
\begin{proof}
We need to prove that if $\rho _1, \rho_2$ are density matrices and $0\le \lambda\le 1$, then
\begin{eqnarray}\label{33b}
{\rm Tr}\{{\cal C}[\lambda \rho _1+(1-\lambda)\rho _2]\}\le \lambda {\rm Tr}[{\cal C}(\rho _1)]+(1-\lambda){\rm Tr}[{\cal C}(\rho_2)].
\end{eqnarray}
In this proof we use the more detailed notation $p [\sigma(i); \rho]$ which indicates explicity the density matrix used in the calculation of $p [\sigma(i)]$.

We start from the relation
\begin{eqnarray}
p [\sigma(i);\lambda \rho _1+(1-\lambda)\rho _2]=\lambda {\rm Tr}\{\rho _1[\Pi[\sigma (i)]\}+(1-\lambda){\rm Tr}\{\rho _2[\Pi[\sigma (i)]\},
\end{eqnarray}
where $\sigma$ is the ranking permutation of $\lambda \rho _1+(1-\lambda)\rho _2$.
From this we get
\begin{eqnarray}
\sum_{i=n-d+1}^{n}p [\sigma(i);\lambda \rho _1+(1-\lambda)\rho _2]= \lambda \sum_{i=n-d+1}^{n} p[\sigma (i);\rho _1]+(1-\lambda) \sum_{i=n-d+1}^{n} p[\sigma (i)];\rho _2].
\end{eqnarray}
Then
\begin{eqnarray}
\sum_{i=n-d+1}^{n}p [\sigma(i);\lambda \rho _1+(1-\lambda)\rho _2]={\rm Tr}\{{\cal C}[\lambda \rho _1+(1-\lambda)\rho _2]\}.
\end{eqnarray}
Also 
\begin{eqnarray}
&&\sum_{i=n-d+1}^{n} p[\sigma (i)];\rho _1\}\le{\rm Tr}[{\cal C}(\rho _1)]\nonumber\\
&&\sum_{i=n-d+1}^{n} p[\sigma (i)];\rho _2\}\le {\rm Tr}[{\cal C}(\rho_2)]
\end{eqnarray}
because the ranking permutations for  the density matrices $\rho_1, \rho _2$ are in general different from $\sigma$.
Therefore these two sums do not contain the $d$ largest probabilities corresponding to $\rho _1, \rho _2$. This proves the proposition.
\end{proof}

\subsection{ The Choquet integral in terms of M\"obius operators}\label{MMM}

The M\"obius transform of Eqs(\ref{M}), in the present context is
\begin{eqnarray}\label{m}
{\mathfrak D} (B)=\sum _{A\subseteq B} (-1)^{|A|-|B|}\Pi(A).
\end{eqnarray}
We refer to ${\mathfrak D} (B)$ as the M\"obius operators.
An example is
\begin{eqnarray}
{\mathfrak D} (\{i,j\})=\Pi(\{i,j\})-\Pi(\{i\})-\Pi(\{j\}).
\end{eqnarray}
We can prove that the ${\mathfrak D} (\{i,j\})$ is related to the commutator $[\Pi(\{i\}),\Pi(\{j\}]$ \cite{VO2,VO3}:
\begin{eqnarray}
[\Pi(\{i\}),\Pi(\{j\}]={\mathfrak D} (\{i,j\})[\Pi(\{i\})-\Pi(\{j\})].
\end{eqnarray}
Therefore the M\"obius operators are related to non-commutativity.
If $[\Pi(\{i\}),\Pi(\{j\}]=0$ then ${\mathfrak D} (\{i,j\})=0$.

In the special case that the states in Eq.(\ref{sig}) form an orthonormal basis, then ${\mathfrak D} (B)=0$ for $|B|\ge 2$, and Eqs.(\ref{2A}), (\ref{377}) become equalities.

The inverse M\"obius transform is
\begin{eqnarray}\label{al}
\Pi (A)=\sum _{B\subseteq A}{\mathfrak D} (B).
\end{eqnarray}
In Eq.(\ref{al}) we put $A=\Omega$ (in which case $\Pi(\Omega)={\bf 1}$), and we get
\begin{eqnarray}\label{al1}
\sum _{i=1}^{n}\Pi(\{i\})+
\sum _{i,j}{\mathfrak D} (\{i,j\})+...+{\mathfrak D}(\{1,...,n\})={\bf 1}.
\end{eqnarray}
The projectors $\Pi (\{i\})$ and $\Pi (\{j\})$ (with $i\ne j$) overlap with each other, in the sense that  $\Pi (\{i\})\Pi (\{j\})\ne 0$. 
The inverse M\"obius transform in Eq.(\ref{al1}) involves the $n$ projectors $\Pi(i)$, 
and all the M\"obius ${\mathfrak D}$-operators, whose role is to remove the overlaps between the $\Pi(i)$ so there is no double-counting.

The following proposition is analogous to proposition \ref{PRO12} and is given without proof \cite{C1}.
It expresses the Choquet integral in terms of  M\"obius operators.

\begin{proposition}\label{www}
${\cal C}(\rho)$ can be written as
\begin{eqnarray}\label{n9b}
&&{\cal C}(\rho)={\cal C}_{1}(\rho)+{\cal C}_{2}(\rho)+...+{\cal C}_{n}(\rho)\nonumber\\
&&{\cal C}_{1}(\rho)=\sum \Pi\{(i\}) p(i)\nonumber\\
&&{\cal C}_{2}(\rho)=\sum {\mathfrak D}(\{i,j\})\min \{p(i),p(j)\}\nonumber\\
&&{\cal C}_{3}(\rho)=\sum {\mathfrak D}(\{i,j,k\})\min \{p(i),p(j), p(k)\}\nonumber\\
&&...............................................................\nonumber\\
&&{\cal C}_{n}(\rho)={\mathfrak D}(\{1,...,n\})p[\sigma(1)]
\end{eqnarray}
\end{proposition}
The projectors $\Pi\{i\})$ are not orthogonal. 
The term ${\mathfrak D}(i,j)\min \{p(i), p(j)\}$ is a `correction' related to the non-orthogonality of the states $\ket{i}, \ket{j}$.
The term ${\mathfrak D}(i,j,k)\min \{p(i), p(j), p(k) \}$ 
is a `correction' related to the non-orthogonality of the states $\ket{i}, \ket{j}, \ket{k}$, etc.

\section{Comonotonic density matrices: physically similar density matrices}\label{L}

\begin{definition}
The density matrices $\rho_1, \rho _2$ are comonotonic,  if they have the same ranking permutation.
We denote this as $\rho _1\overset{c}{\sim} \rho _2$.
\end{definition}
Comonotonic density matrices  `live' in the same postcode $(\sigma(n),...,\sigma(1))$ within the Hilbert space, and in this sense they are physically similar.

It is easily seen that: 
\begin{itemize}

\item
$\frac{1}{d}{\bf 1}\overset{c}{\sim} \rho $ for any density matrix $\rho$. 
\item
For $\lambda \ge 0$, the $\rho$ and $\lambda \rho$ are comonotonic (the latter is a non-normalized density matrix).
\item
If $\rho _1\overset{c}{\sim} \rho _2$ and $0\le \lambda \le 1$, then $\rho _1\overset{c}{\sim}  \rho _1 +(1-\lambda)\rho_2$ and  $\rho _2\overset{c}{\sim}  \rho _1 +(1-\lambda)\rho_2$.
\item
If  $\rho _1\overset{c}{\sim}  \rho_2$, then  the projectors $\varpi _\sigma (i)$ are the same for both $\rho _1, \rho _2$.  Consequently 
\begin{eqnarray}
{\cal C}(\rho_1)=\sum_{i=n-d+1}^{n}p _1[\sigma(i)]\varpi _\sigma (i);\;\;\;
{\cal C}(\rho_2)=\sum_{i=n-d+1}^{n}p _2[\sigma(i)]\varpi _\sigma (i)
\end{eqnarray}
and the ${\cal C}(\rho_1), {\cal C}(\rho_2)$ commute.
\end{itemize}

\begin{proposition}
If $\rho_1, \rho_2$ are comonotonic density matrices and $0\le \lambda\le 1$, then 
\begin{eqnarray}\label{VV}
{\cal C}[\lambda \rho _1 +(1-\lambda)\rho_2]=\lambda{\cal C}(\rho_1)+(1-\lambda){\cal C}(\rho _2).
\end{eqnarray}
\end{proposition}
\begin{proof}
Comonotonic density matrices have the same ranking permutation, and consequently the 
projectors $\varpi _\sigma (i)$ are the same.
The  $\rho_1, \rho_2, \lambda \rho _1 +(1-\lambda)\rho_2$ are pairwise comonotonic, and then we easily prove Eq.(\ref{VV}).
\end{proof}

\subsection{The equivalence classes  ${\cal M}_d^\prime (\sigma)$ }
Comonotonicity is not transitive in the set ${\cal M}_d$. 
For example, $\frac{1}{d}{\bf 1}$ is comonotonic to every density matrix and yet there are density matrices which are not comonotonic.
As in section \ref{nnmm} we define a subset of ${\cal M}_d$ where comonotonicity is transitive.
\begin{definition}
${\cal M}_d^\prime$ is a subset of ${\cal M}_d$ which contains density matrices for which the $p(i)={\rm Tr}[\rho \Pi (\{i\})]$ are different from each other
(if $i\ne j$ then $p(i)\ne p(j)$).
\end{definition}
Comonotonicity is transitive in ${\cal M}_d^\prime$. In this case  we have a strict inequality in Eq.(\ref{56}).
This is analogous to commutativity which is not transitive in general, but it is transitive if we restrict ourselves to matrices with eigenvalues which are different from each other.

Since transitivity holds, comonotonicity is an equivalence relation in ${\cal M}_d^\prime$, which partitions it into equivalence classes
${\cal M}_d^\prime (\sigma)$ (where $\sigma$ is the ranking permutation).
All density matrices within a given equivalence class have the same ranking premutation and the same set of discrete derivatives of the cumulative projectors $\{\varpi _\sigma (i)\;|\;i=n-d+1,...,d\}$.
The set ${\cal M}_d \setminus{\cal M}_d^\prime $ contains the `boundary density matrices' (with equalities in Eq.(\ref{56})).

The number of equivalence classes ${\cal M}_d^\prime (\sigma)$ is
\begin{eqnarray}
{\mathfrak E}(n,d)=\frac{n !}{(n-d)!}.
\end{eqnarray}
This is because only the $d$ largest values of $p [\sigma(i)]$ enter in Eq.(\ref{010}).
In contrast, all $n$ values of $f(i)$ enter in Eq.(\ref{87}) and the number of equivalence classes in ${\cal L}_n^\prime$ is $n!$.

Clearly  $\rho _1\overset{p}{\sim}  \rho_2$ implies that  $\rho _1\overset{c}{\sim}  \rho_2$,
and each of the equivalence classes ${\cal M}_d^\prime (\sigma)$ is the union of an infinite number of the equivalence classes ${\mathfrak M}_d[p(i)]$.

\subsection{A total preorder within ${\cal M}_d^\prime (\sigma)$}

Various quantities can be used to define  total preorders in the set of density matrices.
An example is the entropy $E(\rho)=-{\rm Tr}(\rho \log \rho)$ which defines the total preorder `more mixed'.

Here we define a total preorder  based on the ${\rm Tr}[{\cal C}(\rho )]$.
A preorder is useful only if it has strong properties.
The preorder defined below has strong properties (Eqs. (\ref{K3}), (\ref{K1})) for comonotonic density matrices.
For this reason, we define the  total preorder $\overset{\sigma} \succ$, for density matrices which belong in the same equivalence class ${\cal M}_d^\prime (\sigma)$.
\begin{definition}
Let $\rho _1, \rho _2$ be density matrices in the same equivalence class ${\cal M}_d^\prime (\sigma)$.
Then $\rho _1 \overset{\sigma} \succ \rho _2$ if 
\begin{eqnarray}
{\cal E}_{\rm max}(\sigma)\ge {\rm Tr}[{\cal C}(\rho _1)]\ge {\rm Tr}[{\cal C}(\rho _2)]\ge {\cal E}_{\rm min}(\sigma).
\end{eqnarray}
\end{definition}

\begin{proposition}
$ \overset{\sigma}\succ$ is a total preorder in the set ${\cal M}_d^\prime (\sigma)$.
\end{proposition}
\begin{proof}
Clearly $\overset {\sigma} \succ$ is reflexive ($\rho\overset{\sigma}\succ \rho$)  and transitive (if $\rho _1\overset{\sigma}\succ \rho _2$ and  $\rho _2\overset{\sigma}\succ \rho _3$ then  $\rho _1\overset{\sigma}\succ \rho _3$).
But the antisymmetry property does not hold ($\rho _1 \overset{\sigma}\succ \rho _2$ and $\rho _1 \overset{\sigma}\prec \rho _2$ 
implies that ${\rm Tr}[{\cal C}(\rho _1)]={\rm Tr}[{\cal C}(\rho _2)]$, but it does not follow that $\rho _1 = \rho _2$ ).
Therefore $\overset{\sigma}\succ$ is a preorder, rather than a partial order.
It is a total preorder because for any $\rho _1, \rho _2$, either $\rho _1 \overset{\sigma}\prec  \rho _2$ or $\rho _1\overset{\sigma}\succ \rho _2$.
\end{proof}

The following proposition shows that addition preserves the $\overset{\sigma}\succ$ preorder:
\begin{proposition}\label{cc1}
Let $\rho _1, \rho _2, \rho _3$ be density matrices in the same equivalence class ${\cal M}_d^\prime (\sigma)$,  and $0\le \lambda \le 1$. Then
\begin{eqnarray}\label{K3}
\rho _1\overset{\sigma}\succ \rho _2\;\;\rightarrow\;\; \lambda \rho _1+(1-\lambda) \rho _3\overset{\sigma}\succ  \lambda \rho _2 +(1-\lambda)\rho _3.
\end{eqnarray}
\end{proposition}
\begin{proof}
We have
\begin{eqnarray}
\rho _1\overset{\sigma}\succ \rho _2\;\;\rightarrow\;\;{\rm Tr}[{\cal C}(\rho _1)]\ge {\rm Tr}[{\cal C}(\rho _2)]
\rightarrow\;\;{\rm Tr}[\lambda {\cal C}(\rho _1)+(1-\lambda){\cal C}(\rho _3)]\ge {\rm Tr}[\lambda{\cal C}(\rho _2)+(1-\lambda){\cal C}(\rho _3)]
\end{eqnarray}
Using the additivity of the Choquet integral for comonotonic operators, we rewrite this as
\begin{eqnarray}
{\rm Tr}[{\cal C}(\lambda \rho _1+(1-\lambda)\rho _3)]\ge {\rm Tr}[{\cal C}(\lambda\rho _2+(1-\lambda)\rho _3)],
\end{eqnarray}
and this proves the proposition.

\end{proof}

\begin{proposition}
Let $\rho _1 \overset{\sigma}\succ \rho _2$ and $0\le \lambda \le 1$. Then
\begin{eqnarray}\label{K1}
\rho _1 \overset{\sigma}\succ \lambda \rho _1+(1-\lambda) \rho _2\overset{\sigma}\succ \rho _2.
\end{eqnarray}
\end{proposition}
\begin{proof}
Using Eq.(\ref{33b}) and the fact that ${\rm Tr}[{\cal C}(\rho _1)]\ge {\rm Tr}[{\cal C}(\rho _2)]$, we get
\begin{eqnarray}
{\rm Tr}\{{\cal C}[\lambda \rho _1+(1-\lambda)\rho _2]\}\le \lambda {\rm Tr}[{\cal C}(\rho _1)]+(1-\lambda){\rm Tr}[{\cal C}(\rho_2)]\le {\rm Tr}[{\cal C}(\rho _1)].
\end{eqnarray}
This proves the left hand side.

The right hand side follows from proposition \ref{cc1} with $\rho_2=\rho_3$.
\end{proof}

If $\rho _1 \overset{\sigma}\succ \rho _2$ then $\rho _1, \rho _2$ have the same ranking permutation $\sigma$, and
then ${\rm Tr}{\cal C}(\rho_1)\ge {\rm Tr}{\cal C}(\rho_2)$ means that $\rho _1$ is more aligned to the projectors $\Pi (\{i\})$ than $\rho _2$.
We can call $\overset{\sigma}\succ$ `more aligned' preorder.

\section{Approximation of a partially known density matrix by its Choquet integral}

We consider a  partially known density matrix $\rho$, with given $\{p(i)\}$,  where
$i=1,...,n$ and $d\le n\le d^2-1$.
The
\begin{eqnarray} 
R(\rho) =\frac{{\cal C}(\rho)}{{\rm Tr}[{\cal C}(\rho)]};\;\;\;{\rm Tr}[R(\rho)]=1,
\end{eqnarray}
 is a positive semi-definite Hermitian matrix with trace $1$, i.e., a density matrix.
We show that $R(\rho)$ is a good (in the sense discussed below) approximation to $\rho$.
Using Eq.(\ref{010}) we express $R(\rho)$ as
\begin{eqnarray}
R(\rho)=\sum_{i=n-d+1}^{n}{\widehat p} [\sigma(i)]\varpi _\sigma (i);\;\;\;{\widehat p} [\sigma(i)]=\frac{p [\sigma(i)]}{{\rm Tr}[{\cal C}(\rho)]};\;\;\;
\sum_{i=n-d+1}^{n}{\widehat p} [\sigma(i)]=1.
\end{eqnarray}
If $\Theta$ is an arbitrary operator
\begin{eqnarray}\label{nnn}
{\rm Tr}[R(\rho) \Theta]=\sum_{i=n-d+1}^{n}{\widehat p} [\sigma(i)]{\rm Tr}[\varpi _\sigma (i)\Theta].
\end{eqnarray}
\begin{proposition}\label{PRO34}
Let $\rho$ be a density matrix with ranking permutation $\sigma$ and 
\begin{eqnarray}\label{nnn}
{\mathfrak P}(j)={\rm Tr}[R(\rho) \Pi (j)].
\end{eqnarray}
Then
\begin{eqnarray}\label{nnn}
{\mathfrak P}[\sigma (n)]\ge {\mathfrak P}(i),
\end{eqnarray}
for all $i$.
\end{proposition}
\begin{proof}

We multiply Eq.(\ref{00}) by $\Pi [\{\sigma (j)\}]$ and we get the trace:
\begin{eqnarray}\label{90}
{\rm Tr}\{{\cal C}(\rho))\Pi [\{\sigma (j)\}]\}&=&p [\sigma(n-d+1)]){\rm Tr}\{\Pi [\{\sigma (j)\}]\}\nonumber\\&+&\{p[\sigma (n-d+2)]-p [\sigma (n-d+1)]\}{\rm Tr}\{\pi_\sigma(n-d+2))\Pi [\{\sigma (j)\}]\}+....\nonumber\\
&+&\{p[\sigma (n)]-p [\sigma (n-1)]\}{\rm Tr}\{\pi_\sigma (n))\Pi [\{\sigma (j)\}]\}
\end{eqnarray}
Since
\begin{eqnarray}
{\rm Tr}\{\Pi [\{\sigma (j)\}]\}=1;\;\;\;p[\sigma (i+1)]-p [\sigma (i)]\ge 0\;\;\;\;0\le {\rm Tr}\{\pi_\sigma(i)\Pi [\{\sigma (j)\}]\}\le 1,
\end{eqnarray}
we get
\begin{eqnarray}\label{PP1}
{\rm Tr}\{{\cal C}(\rho))\Pi [\{\sigma (j)\}]\}&\le &p [\sigma(n-d+1)])+\{p[\sigma (n-d+2)]-p [\sigma (n-d+1)]\}+....\nonumber\\
&+&\{p[\sigma (n)]-p [\sigma (n-1)]\}= p[\sigma (n)]
\end{eqnarray}
For $j=n$ this inequality becomes equality. Indeed,  $\pi_\sigma(i)\Pi [\{\sigma (n)\}]=\Pi [\{\sigma (n)\}]$ and  Eq.(\ref{90}) with $j=n$  gives
\begin{eqnarray}\label{PP2}
{\rm Tr}\{{\cal C}(\rho)\Pi [\{\sigma (n)\}]\}= p[\sigma (n)]
\end{eqnarray}
Combining Eqs(\ref{PP1}), (\ref{PP2})  we get ${\mathfrak P}[\sigma (n)]\ge {\mathfrak P}(i)$.

\end{proof}

The above proposition shows that from all ${\mathfrak P}(i)$ the ${\mathfrak P}[\sigma (n)]$ has the maximum value.
Similarly from all $p(i)$ the $p[\sigma (n)]$ has the maximum value.
Therefore both $\rho$ and $R(\rho)$ `live' primarily in  the same subspace $H[\{\sigma (n)\}]$ and to a lesser extend in the other subspaces $H(\{i\})$.
In this case we say that $\rho$ and $R(\rho)$ are weakly comonotonic.
\begin{definition}
Let  $\rho_1, \rho _2$ be density matrices with ranking permutations $\sigma _1$ and $\sigma _2$, correspondingly.
 $\rho_1, \rho _2$ are weakly comonotonic if $\sigma _1(n)=\sigma_2(n)$.
\end{definition}
For comonotonic density matrices we have $\sigma _1(i)=\sigma_2(i)$ for all $i$, and therefore weak comonotonicity is a weaker concept than comonotonicity.

$R(\rho)$ is a good approximation of $\rho$, in the sense that they are weakly comonotonic to each other.
The quantity
\begin{eqnarray}\label{error}
{\mathfrak E}=\frac{1}{n}\sum _{j=1}^n[{\mathfrak P}(j)-p(j)]=\frac{1}{n}\sum _{j=1}^n{\rm Tr}\{[R(\rho)-\rho] \Pi (j)\}
\end{eqnarray}
is the average value of ${\mathfrak P}(j)-p(j)$, and it
can be used as an indicator of the error of the approximation.
We compare here the $p(j)$ (which is all the information we have about $\rho$) with the corresponding quantities ${\mathfrak P}(j)$ for $R(\rho)$ (which is the approximation that we use for $\rho$). 

\subsection{Example}\label{EXA89}

In the $3$-dimensional Hilbert space $H(3)$, we consider the following $n=4$ vectors:
\begin{eqnarray}\label{bn}
\ket{1}=
\begin{pmatrix}
1\\
0\\
0\\
\end{pmatrix};\;\;\;
\ket{2}=\frac{1}{\sqrt 2}
\begin{pmatrix}
1\\
0\\
1\\
\end{pmatrix};\;\;\;
\ket{3}=\frac{1}{\sqrt 2}
\begin{pmatrix}
0\\
1\\
1\\
\end{pmatrix};\;\;\;
\ket{4}=\frac{1}{3}
\begin{pmatrix}
2\\
1\\
2\\
\end{pmatrix},
\end{eqnarray}
In this case
\begin{eqnarray}\label{PR1}
\Pi (\{1\})=
\begin{pmatrix}
1&0&0\\
0&0&0\\
0&0&0\\
\end{pmatrix};\;\;
\Pi (\{2\})=\frac{1}{2}
\begin{pmatrix}
1&0&1\\
0&0&0\\
1&0&1\\
\end{pmatrix};\;\;
\Pi (\{3\})=\frac{1}{2}
\begin{pmatrix}
0&0&0\\
0&1&1\\
0&1&1\\
\end{pmatrix};\;\;
\Pi (\{4\})=\frac{1}{9}
\begin{pmatrix}
4&2&4\\
2&1&2\\
4&2&4\\
\end{pmatrix}.
\end{eqnarray}
The minimum and maximum  eigenvalues of $Q=\frac{1}{4}[\Pi (\{1\})+\Pi (\{2\})+\Pi (\{3\})+\Pi (\{4\})]$, are ${\cal E}_{\rm min}=0.04$ and ${\cal E}_{\rm max}=0.71$.

We assume that measurements with these projectors on a density matrix $\rho$ give
\begin{eqnarray}
p(1)=0.2;\;\;\;p(2)=0.7;\;\;\;p(3)=0.4;\;\;\;p(4)=0.5
\end{eqnarray}
These values obey the contraint in Eq.(\ref{BB}).
As we mentioned earlier, different ensembles of the same density matrix are used in these measurements (because the projectors do not commute).

We express the density matrix $\rho$ as
\begin{eqnarray}\label{P1}
\rho =
\begin{pmatrix}
\rho _1&a_1+ib_1&a_2+ib_2\\
a_1-ib_1&\rho_2&a_3+ib_3\\
a_2-ib_2&a_3-ib_3&1-\rho_1-\rho_2\\
\end{pmatrix}
\end{eqnarray}
in terms of  $8$ real variables, and we get
\begin{eqnarray}
&&p(1)=\rho_1=0.2\nonumber\\
&&p(2)=\frac{1}{2}(1-\rho_2)+a_2=0.7\nonumber\\
&&p(3)=\frac{1}{2}(1-\rho_1)+a_3=0.4\nonumber\\
&&p(4)=\frac{1}{9}(4-3\rho_2+4a_1+8a_2+4a_3)=0.5
\end{eqnarray}
Therefore the $\rho _2, b_1, b_2, b_3$ take arbitrary real values, and  
\begin{eqnarray}\label{P2}
\rho_1=0.2;\;\;\;a_1=-\frac{1.1+\rho_2}{4};\;\;\;a_2=\frac{0.4+\rho_2}{2};\;\;\;a_3=0.
\end{eqnarray}
Extra constraints are required for $\rho$ to be positive-semidefinite matrix, and it is not always easy to find these constraints.
In this paper we propose that the Choquet integral (which is a positive semi-definite matrix) can be used as an approximation to these density matrices.

In this example $p(1)<p(3)<p(4)<p(2)$, and the ranking permutation is 
\begin{eqnarray}
\sigma (1)=1;\;\;\sigma (2)=3;\;\;\sigma (3)=4;\;\;\sigma (4)=2.
\end{eqnarray}
The cumulative projectors related to this order are
\begin{eqnarray}
\pi_\sigma (4)=\Pi(\{2\});\;\;
\pi_\sigma(3)= \Pi[\{4,2\}];\;\;
\pi_\sigma(2)= {\bf 1};\;\;
\pi_\sigma(1)= {\bf 1}.
\end{eqnarray}
Using Eq.(\ref{670A}), we find
\begin{eqnarray}\label{PR2}
&&\Pi (\{4,2\})=\frac{1}{2}
\begin{pmatrix}
1&0&1\\
0&2&0\\
1&0&1\\
\end{pmatrix}.
\end{eqnarray}
Consequently
\begin{eqnarray}
{\cal C}(\rho)&=&p(3){\bf 1}+[p(4)-p(3)]\Pi(\{4,2\})+[p(2)-p(4)]\Pi(\{2\})=
\begin{pmatrix}
0.55&0&0.15\\
0&0.50&0\\
0.15&0&0.55\\
\end{pmatrix}
\end{eqnarray}
and 
${\rm Tr}[{\cal C}(\rho)]=1.6$. In this case the minimum and maximum  eigenvalues of $Q_{\sigma}=\frac{1}{3}[\Pi (\{2\})+\Pi (\{3\})+\Pi (\{4\})]$ are
${\cal E}_{\rm min}=0.007$ and ${\cal E}_{\rm max}=0.816$ and obey the constraint in Eq.(\ref{BBB}).

From this we get
\begin{eqnarray}
R(\rho)=
\begin{pmatrix}
0.34&0&0.09\\
0&0.32&0\\
0.09&0&0.34\\
\end{pmatrix}.
\end{eqnarray}
Using this density matrix we find
\begin{eqnarray}
{\mathfrak P}(1)=0.34;\;\;\; {\mathfrak P}(2)=0.43;\;\;\; {\mathfrak P}(3)=0.32;\;\;\; {\mathfrak P}(4)=0.42
\end{eqnarray}
and calculate the error ${\mathfrak E}=-0.07$.

We note that among the $\{p(i)\}$ the $p(2)$ is the largest, and among the   $\{{\mathfrak P}(i)\}$ the  ${\mathfrak P}(2)$ is the largest.
Therefore $\rho$ and $R(\rho)$ are weakly comonotonic to each other. They are both closer to  the same subspace $H(\{2\})$ (related to the vector $\ket{2}$ in Eq.(\ref{bn})) and to a lesser extend to the other subspaces 
$H(\{i\})$.

\section{Discussion}

In many cases the additivity property of Kolmogorov probabilities in Eq.(\ref{B}) does not hold.
The weaker concept of capacity or non-additive probability obeys Eq.(\ref{A}).
Additivity is the starting point of standard integration theory, and in the case of capacities is replaced by the Choquet integration.

Choquet integrals in a classical context are defined  in Eq.(\ref{87}), in terms of cumulative functions and their derivatives.
In general ${\cal C}(f+g;\mu)\ne {\cal C}(f;\mu)+{\cal C}(g;\mu)$, and only in the case of comonotonic functions this becomes equality.
The Choquet integral  ${\cal C}(f;\mu)$ is expressed in terms of M\"obius transforms of the capacities, in section \ref{MOB}.
An example is given in section \ref{EXA23}.

In a quantum context the non-additivity is expressed in terms of non-orthogonal projectors (which are intimately related to probabilities) in Eq.(\ref{2A}).
Cumulative projectors and their discrete derivatives, are studied in section \ref{ASD}, and they are used to define Choquet integrals in Eqs.(\ref{010}),(\ref{00}).
Here the Choquet integrals ${\cal C}(\rho)$ are positive semi-definite matrices.
Upper and lower bounds for the ${\rm Tr}[{\cal C}(\rho)]$ are given in Eq.(\ref{BBB}).
It has also been shown that  ${\rm Tr}[{\cal C}(\rho)]$ is a convex function.
The Choquet integral  ${\cal C}(\rho)$ is expressed in terms of M\"obius operators, in section \ref{MMM}.

In general ${\cal C}[\lambda \rho _1 +(1-\lambda)\rho_2]$ is different from $\lambda{\cal C}(\rho_1)+(1-\lambda){\cal C}(\rho _2)$.
Only for comonotonic density matrices (which can be interpreted as physically similar density matrices) this becomes equality.
Comonotonicity splits the set of density matrices into equivalence classes, and within each of these classes we 
define a total preorder.

As an application we consider a partially known density matrix $\rho$  with given 
$p(i)={\rm Tr}[\rho \Pi(i)]$  (where $i=1,...,n$ and $d\le n\le d^2-1$).
The ${\cal C}(\rho)/{\rm Tr}[{\cal C}(\rho)]$ is a good approximation to $\rho$,
in the sense that they are weakly comonotonic to each other.
The ${\mathfrak E}$ in Eq.(\ref{error}) quantifies the error in this approximation. 

The work uses non-orthogonal total sets of states (Eq.(\ref{sig})) for the study of physical problems.

\end{document}